\documentclass[english,11pt,reqno]{amsart}

\usepackage{graphicx}
\usepackage{hyperref}
\usepackage{subcaption}
\usepackage{amsmath}
\usepackage{amssymb}
\usepackage{natbib}

\def\di{\displaystyle}

\theoremstyle{plain}
\newtheorem{theorem}{Theorem}[section]

\newtheorem{lemma}[theorem]{Lemma}
\newtheorem{proposition}[theorem]{Proposition}

\theoremstyle{definition}
\newtheorem{definition}[theorem]{Definition}

\theoremstyle{remark}
\newtheorem{remark}{Remark}

\begin{document}

\title{Dynamics of a rotating ellipsoid with a stochastic flattening}

\author{Etienne Behar$^1$, Jacky Cresson$^{1,2}$ and Fr\'ed\'eric Pierret$^1$}

\begin{abstract}
Experimental data suggest that the Earth short time dynamics is related to stochastic fluctuation of its shape. As a first approach to this problem, we derive a toy-model for the motion of a rotating ellipsoid in the framework of stochastic differential equations. Precisely, we assume that the fluctuations of the geometric flattening can be modeled by an admissible class of diffusion processes respecting some invariance properties. This model allows us to determine an explicit drift component in the dynamical flattening and the second zonal harmonic whose origin comes from the stochastic term and is responsible for short term effects. Using appropriate numerical scheme, we perform numerical simulations showing the role of the stochastic perturbation on the short term dynamics. Our toy-model exhibits behaviors which look like the experimental one. This suggests to extend our strategy with a more elaborated model for the deterministic part.
\end{abstract}

\maketitle

\vskip 5mm
\begin{tiny}
	\begin{enumerate}
		\item {Laboratoire de Math\'ematiques Appliqu\'ees de Pau, Universit\'e de Pau et des Pays de l'Adour, avenue de l'Universit\'e, BP 1155, 64013 Pau Cedex, France}
		
		\item {SYRTE UMR CNRS 8630, Observatoire de Paris and University Paris VI, France}
	\end{enumerate}
\end{tiny}
\vskip 5mm

\tableofcontents

\section{Introduction}

The irregularities in the Earth's rotation axis direction and norm originate in various complex phenomena such as interaction with Solar system bodies, mass redistribution in the oceans and the atmosphere, as well as interactions between the various internal layers (see \cite{lambeck_1989}, Figure \ref{fig:poleyear} and Figure \ref{fig:j2}).\\

\begin{figure}[ht!]
	\begin{subfigure}[t]{0.45\textwidth}
		\centering
		\includegraphics[height=0.8 \textwidth]{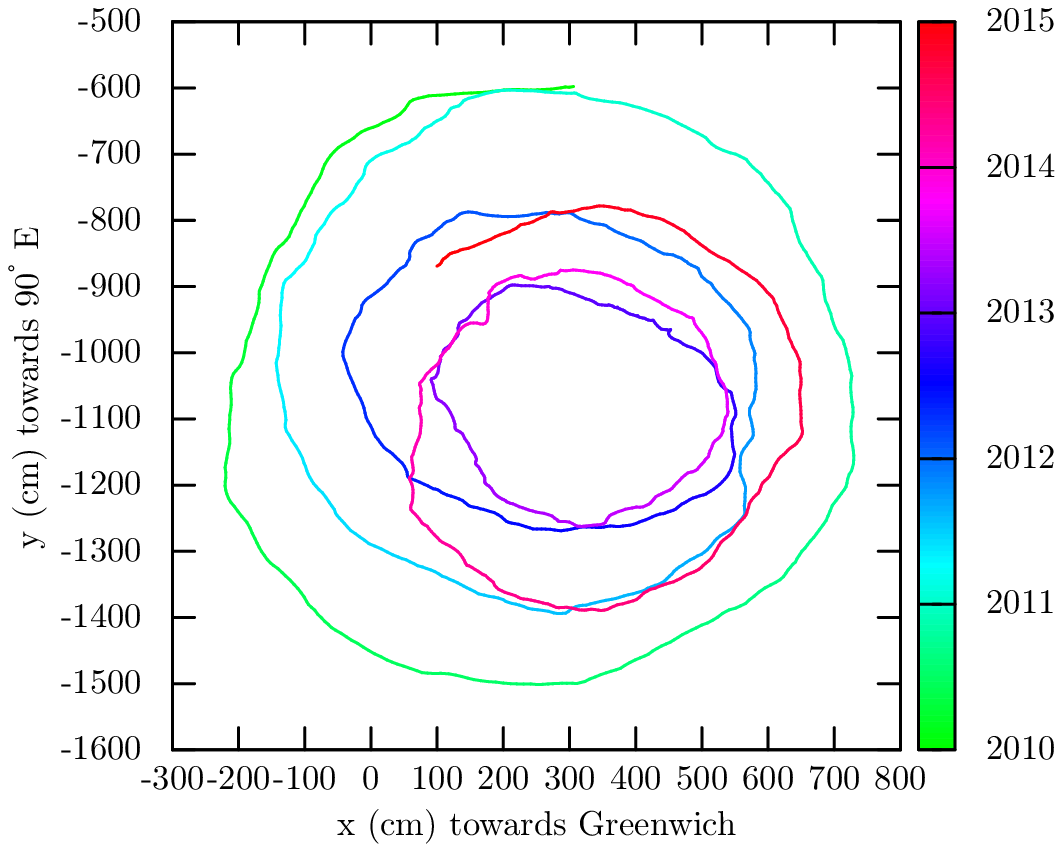}
		\caption{Evolution of the position of the Earth's pole over few years, \tiny http://hpiers.obspm.fr/eop-pc/}
		\label{fig:poleyear}
	\end{subfigure}
	\quad 
	\begin{subfigure}[t]{0.45\textwidth}
	\centering
	\includegraphics[height=0.8\textwidth]{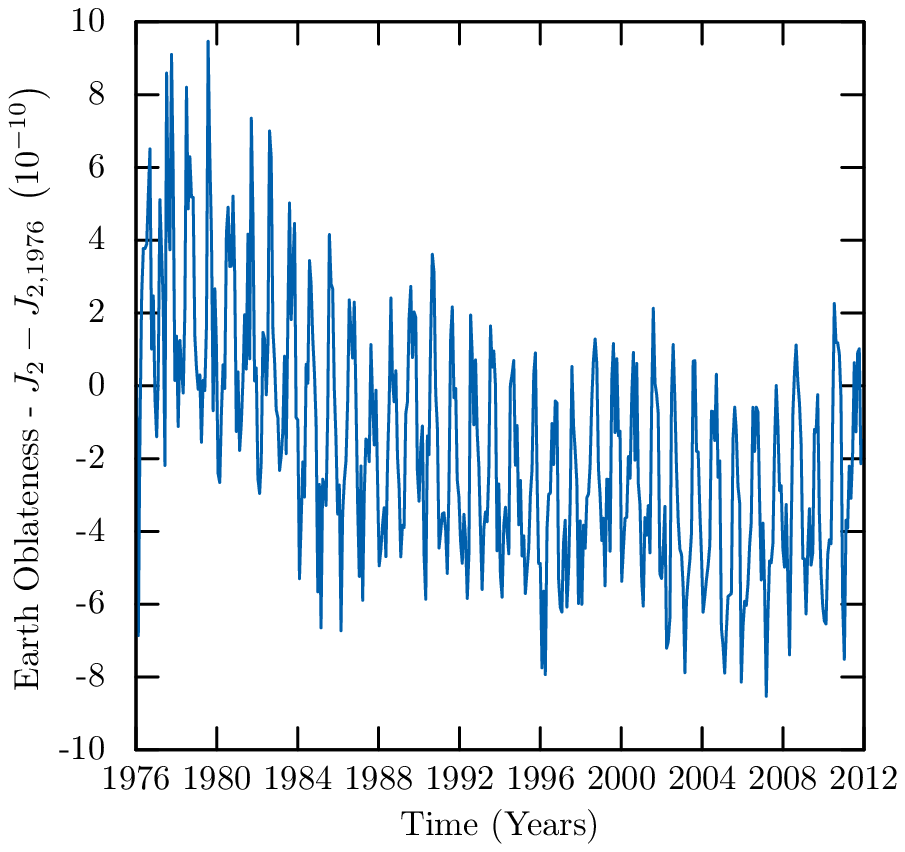}
	\caption{Earth's second degree zonal harmonic $J_2$, \tiny http://grace.jpl.nasa.gov/data/J2/}
	\label{fig:j2}
	\end{subfigure}
	\caption{Dynamical illustrations of the Earth.}
\end{figure}

In order to take into account these irregularities, classical models of Earth's dynamics are constructed on geophysical considerations such as oceans and atmosphere dynamics and the Earth's neighborhood like the Moon, the planets and the Sun (see \cite{bizouard}, \cite{Jin20131}, \cite{chao1993}, \cite{barnes1983}, \cite{lambeck_1989},  \cite{sidorenkov2009}). In these models, the short time (diurnal and subdiurnal time-scale) irregularities are badly modeled due to the complexity of the phenomena (see \cite{viron2005}, \cite{yseboodt2002}, \cite{cheng_variation_obl}). With this problem, we are led to the following question : \textit{Is an alternative approach of Earth's rotation model possible ?} \\

The complex mechanisms underlying the irregularities in the Earth's rotation strongly suggest modeling the Earth's rotation, over short time scales, with random processes. Indeed, the short time variations of the \emph{rotation speed}, \emph{length of the day} and \emph{polar motion} are strongly correlated with the time variations in the dynamics of the ocean and the atmosphere, for periods of order between the day and the year. It has been observed that these short time variations seem to be of stochastic nature (see \cite{eubanks1988}, \cite{lambeck_1989}, \cite{sidorenkov2009}). This induces strong changes in the modeling process. An example of such considerations is the two-body problem with a stochastic perturbation studied in \cite{cpp}. Up to now, the stochastic behavior has been taken into account using \emph{filtering theory} which consists in adding noises governed by constants and adjust them to best estimate the observations (see \cite{hamdan1996},\cite{markov2005}, \cite{chin2005}, \cite{chin2009}). Such a method, although effective, can not be used to determine how a given stochastic perturbation impacts the other quantities of interest. Indeed, the form of the stochastic process is not explicitly related to the physical parameters entering in the physical process as the experimental data mix phenomenon of different origins. \\

In this work, we model an oblate homogeneous ellipsoid of revolution, which could represent the Earth, whose geometric flattening is varying and contains a stochastic component. We use the framework of stochastic differential equations in the sense of It\^o. The major difference with an approach through filtering theory is that we are looking for a physical phenomenon linked to the ellipsoid itself by explicit formula. This allows us to identify the influence of the stochastic term on the stochastic variation of the flattening, which is responsible of the stochastic behavior observed in the zonal harmonic $J_2$ and the length of the day. Precisely, we obtain explicit formulas relating the stochastic variation of the geometric flattening to the stochastic fluctuation of the dynamical flattening and the second zonal harmonic (see Proposition \ref{H_drift} and Lemma \ref{j2_drift}). This result is consistent with the expected interactions between these different phenomenon, in particular for what concerns the length of the day.\\

It must be noted that adding a stochastic contribution to a known deterministic model is not easy. At least two difficulties must be overcome :
\begin{itemize}
\item First, one must find expressions of quantities of interest which can be computed in the stochastic setting. In our case, we are concerned with adding random fluctuations of the ellipsoid shape to an existing deterministic models. A non exhaustive list of models are given in \cite{barnes1983}, \cite{chao1993}, \cite{vermeersen1999}, \cite{getino1990}, \cite{getino1991}. As a first approach, we restrict our attention to the Euler-Liouville equation (see Section \ref{det_eq_motion}).
\item Second, one must be careful with the stochastic component entering in the geometric flattening. Indeed, without any assumptions, a stochastic process induces unbounded fluctuations leading to unrealistic values.  Then, one must construct an ``admissible'' stochastic deformation having bounded variations with probability one (see Section \ref{adm_stoc_def}).
\end{itemize}

Another difficulty deserves to be mentioned and concerns the numerical study of such kind of models. Indeed, classical numerical schemes do not preserve in general the specific constraints of a model. For example, the usual Euler-Maruyama scheme destroys the invariance condition used to construct admissible stochastic deformations, leading to inconsistent results, even for a short time simulation. This can be overcome using an appropriate time step during the numerical integration (see Section \ref{num_inv} and \cite{pierret_nsem}).\\

The plan of this paper is as follows : \\

In Section 2, we remind the classical equations of motion for a rigid ellipsoid. Section 3 deals with the case of an ellipsoid with a time variable flattening : deterministic or stochastic. In particular we discuss the notion of admissible deformations based on the invariance criterion for (stochastic) differential equations. Section 4 is devoted to the numerical exploration of a toy-model obtained by a particular deformation equation of the flattening. In Section 5 we conclude and give some perspectives.

\section{Free motion of a rigid ellipsoid}

In this section we remind the equations of motion for a rotating homogeneous rigid ellipsoid. We refer to Chapter 4 and 5 of \cite{goldstein}, Chapter 6 of \cite{landau} and Chapter 3 of \cite{lambeck1988} for full details.\\

We consider an ellipsoid of revolution $\mathcal{E}$ of major axis $a$ and $c$ of mass $M_\mathcal{E}$ and volume $V_\mathcal{E}$. Let \textbf{L} be the angular momentum of $\mathcal{E}$ with $\mathbf{L}=\mathbf{I}\mathbf{\Omega}$ where \textbf{I} is the inertia matrix of $\mathcal{E}$ and $\mathbf{\Omega}$ is the rotation vector. The equation of free motion for $\mathcal{E}$ is 

\begin{equation}
	\frac{d\mathbf{L}}{dt} + \mathbf{\Omega} \wedge \mathbf{L} = 0 \label{eqmotion}.
\end{equation}

We remind that free motion means that there is no external moments acting on the body $\mathcal{E}$. In the principal axes which are the reference frame attached to the center of $\mathcal{E}$ and where the inertia is diagonal whose coefficients are directly linked to the major axis $a$ and $c$. Indeed, in the case of an ellipsoid of revolution, the inertia matrix is expressed as 

\begin{equation}
	{\bf I} = 
	\left( \begin{array}{ccc}
		I_{1} & 0 & 0 \\
		0 & I_{2} & 0 \\
		0 & 0 & I_{3}
	\end{array} \right)
\end{equation}
where $I_2=I_1$, $I_1 = \frac{1}{5} M_\mathcal{E} (a^2+c^2) \label{I1}$ and $I_3=\frac{2}{5}M_\mathcal{E}a^2 \label{I3}$. The volume satisfies the classical formula $V_\mathcal{E}=\frac{4}{3}\pi a^2 c \label{V}$. In the principal axes, the equation of free motion is expressed as
\begin{equation}
	\begin{array}{lll}
		\frac{d\Omega_1}{dt} & =&  \frac{I_1 - I_3}{I_1}\Omega_3 \Omega_2, \\
		\frac{d\Omega_2}{dt} & = & -\frac{I_1 - I_3}{I_1}\Omega_3 \Omega_1, \\
		\frac{d\Omega_3}{dt} & = & 0 .
	\end{array}
\end{equation}
Theses equations are the well known {\it Euler-Liouville equations} of a body in rotation in the case of an ellipsoid of revolution. \\

Multiple definitions related to the characterization of an oblate homogeneous rigid ellipsoid exist. We remind the three most used (see \cite{bizouard}, Appendix C and \cite{lambeck1988}, Eq. 2.4.6,  2.4.7):
\begin{itemize}
	\item The \emph{geometric flattening} $f$ which is the quantity related to the major-axis as $\displaystyle \frac{a-c}{a}$,
	\item The \emph{dynamical flattening} $H$ which is the quantity related to the inertia coefficients as $\frac{I_3-I_1}{I_3}$,
	\item The \emph{second degree zonal harmonic} $J_2$ which is the quantity related to the inertia coefficients and major-axis as $\frac{I_1-I_3}{Ma^2}$.
\end{itemize} 

The flattening that we consider is a geometric variation, a temporal evolution of his shape. In consequence, we always refer to the \emph{geometric flattening} when we discuss about the \emph{flattening}.

\section{Motion of an ellipsoid with time-varying flattening}

We are interested in variations of the flattening and we want to derive the perturbed Euler-Liouville equation of motion under the following assumptions:
\begin{align}
	&\text{(H1) Conservation of the ellipsoid mass} \ M_\mathcal{E}. \\
	&\text{(H2) Conservation of the ellipsoid volume} \ V_\mathcal{E}, \\
	&\text{(H3) Bounded variation of the flattening}.
\end{align}
Those assumptions are physically consistent with observations and the physical considerations as we are only interested in a first approach by the effect of a homogeneous flattening. \\

The entire dynamic will be encoded and described with the major axis $c_t$ through the formula of the inertia matrix and the volume. The basic idea to approach variation of the flattening is that there exists a "mean" deformation of the flattening and a lower and an upper variation around it. Characterization of admissible deformations under the assumptions (H3) depends on its nature, i.e. deterministic or stochastic.

\subsection{Motion of an ellipsoid with deterministic flattening}

\subsubsection{Deterministic variation of the flattening}

Let $c_t$ satisfying the differential equation
\begin{equation}
	\frac{dc_t}{dt} = f(t,c_t)
\end{equation}
where $f\in \mathcal{C}^2(\mathbb{R}\times\mathbb{R},\mathbb{R})$. \\

{\it Consequence of assumption (H1)} : Computing the derivative of the volume formula (\ref{V})
\begin{equation}
	a_t^2=\frac{3V_\mathcal{E}}{4\pi} \frac{1}{c_t},
\end{equation}
we obtain
\begin{equation}
	\frac{d(a_t^2)}{dt} = \frac{3V_\mathcal{E}}{4\pi} \left( -\frac{1}{c_t^2}\frac{dc_t}{dt} \right).
\end{equation}
Thus using the expression of $\frac{dc_t}{dt}$ we obtain the following lemma :

\begin{lemma}
Under assumption (H1) the variation of $a$ is given by
\begin{equation}
	\frac{d(a_t^2)}{dt} = \frac{3V_\mathcal{E}}{4\pi}\left( -\frac{f(t,c_t)}{c_t^2} \right).
\end{equation}

\end{lemma}
We can now determine the variation of the inertia matrix coefficients $I_1$ and $I_3$. 

{\it Consequence of assumption (H2)} : Computing the derivative of the expression of $I_1$ and $I_3$ gives the following lemma

\begin{lemma}
	Under assumption (H2) the variation of $I_1$ and $I_3$ are given by
	\begin{equation}
		\frac{dI_3}{dt} = \frac{3 M_\mathcal{E} V_\mathcal{E}}{10\pi}\left( -\frac{f(t,c_t)}{c_t^2} \right)
	\end{equation}
	and
	\begin{equation}
		\frac{dI_1}{dt} = \frac{M_\mathcal{E}}{5} \left(-\frac{3V_\mathcal{E}}{4\pi}\frac{f(t,c_t)}{c_t^2} + 2 c_t f(t,c_t)\right).
	\end{equation}
\end{lemma}

\subsubsection{Deterministic equations of motion}
\label{det_eq_motion}
In order to formulate the equations of motion of $\mathcal{E}$ with a deterministic flattening, we first rewrite the equations of motion as
\begin{equation}
	\frac{dL_i}{dt} = l_i(\mathbf{I},\mathbf{\Omega}) ,
\end{equation}
with $l_1({\bf I},\boldsymbol \Omega) = (I_1 - I_3) \Omega_2\Omega_3$, $ l_2({\bf I},\boldsymbol \Omega) =  - (I_1 - I_3) \Omega_1\Omega_3$ and $l_3({\bf I},\boldsymbol \Omega)=0$.

Taking into account our deterministic variation of the flattening, we get the full set of the deterministic equations of motion for $\mathcal{E}$ as
\begin{equation}
	\begin{array}{lll}
		\frac{dL_i}{dt} & = & l_i({\bf I},\boldsymbol \Omega),\\
		\frac{dI_i}{dt} & = & k_i(c_t) \nonumber, 
	\end{array}
\end{equation}
for $i=1,2,3$ where
\begin{equation}
	\begin{array}{lll}
		k_1(c_t)&=&\frac{M_\mathcal{E}}{5} \left(-\frac{3V_\mathcal{E}}{4\pi}\frac{f(t,c_t)}{c_t^2} + 2 c_t f(t,c_t)\right), \\
		k_3(c_t)&=&\frac{3 M_\mathcal{E} V_\mathcal{E}}{10\pi}\left( -\frac{f(t,c_t)}{c_t^2} \right), \\
		k_2(c_t)&=& k_3(c_t).
	\end{array}
\end{equation}
A deterministic version of the Euler-Liouville equation induced by the deterministic flattening can then be obtained. As we consider only variation of the flattening, we still have a rotational symmetry . Hence, we have $L_i = I_i \Omega_i$  or equivalently $\Omega_i = \frac{L_i}{I_i}$ for $i=1,2,3$. Computing the derivative for each component of $\mathbf{\Omega}$ we obtain the following definition :

\begin{definition}
	We call Deterministic Euler-Liouville equations for an ellipsoid with a deterministic flattening the following equations
	\begin{equation}
		\begin{array}{lll} 
			\frac{d\Omega_i}{dt} &= \left(\frac{l_i({\bf I},\boldsymbol \Omega)}{I_i} - \frac{\Omega_i}{I_i^2}k_i(c_t)\right), \\
			\frac{dI_i}{dt} & = k_i(c_t), \\
			\frac{dc_t}{dt} &= f(t,c_t)
		\end{array}
	\end{equation}
	for $i=1,2,3$.
\end{definition}

\subsubsection{Admissible deterministic deformations}

We give the form of the differential equations governing a deformation respecting assumption $(H3)$ in the deterministic case.

\begin{definition}
	Let $d_{min}<0$ and $d_{max}>0$ fixed values which correspond to the minimum and maximum variation with respect to the initial value $c_0 > 0$, with $d_{min} +c_0 >0$. If $c_t$ satisfies the condition $c_0 + d_{min} \leq c_t \leq c_0 + d_{max}$ for $t\ge 0$ then we say that $c_t$ is an admissible deterministic deformation.
\end{definition}

In order to characterize admissible deterministic deformations we use the classical invariance theorem (see \cite{walter}, \cite{pavel}) :

\begin{theorem}\label{invd}
	Let $a,b\in\mathbb{R}$ such that $b>a$ and $\frac{dX(t)}{dt} =  f(t,X(t))$ where $f \in \mathcal{C}^2(\mathbb{R}\times\mathbb{R},\mathbb{R})$. Then, the 
	set 
	$$
	K:=\{x\in\mathbb{R} :\ a\leq x\leq b\}
	$$ 
	is invariant for $X(t)$ if and only if
	\begin{eqnarray}
		f (t,a) &\geq & 0 ,\nonumber \\
		f (t,b) &\leq & 0 ,\nonumber
	\end{eqnarray}
	for all $t\geq 0$.
\end{theorem}

\begin{lemma}[Characterization of admissible deterministic deformations]
	Let $c_t$ satisfying $\frac{dc_t}{dt} =  f(t,c_t)$ then $c_t$ is an admissible deterministic variation if and only if
	\begin{align*}
		f(t,c_0 + d_{min}) & \geq 0, \\
		f(t,c_0 + d_{max}) & \leq 0 \ , \quad \forall t \geq 0.
	\end{align*}
\end{lemma}

\subsubsection{A deterministic toy-model}

In order to perform numerical simulations, we define an ad-hoc admissible deformations given by
\begin{equation}
	f(x) = \alpha\cos(\gamma t)(x-(c_0 + d_{min}))((c_0+d_{max})-x) \ , \quad e,\alpha \in \mathbb{R}^+ .
\end{equation}
where $\alpha$ and $\gamma$ are real numbers. As a consequence, the major axis $c_t$ satisfies the differential equation
\begin{equation}
	\frac{dc_t}{dt} = \alpha\cos(\gamma t)(c_t-(c_0 + d_{min}))((c_0+d_{max})-c_t) .
\end{equation}

\begin{remark}
	It is reasonable to take a periodic deformation for the deterministic part as we observe such kind of variations for the Earth's oblateness (see \cite{bizouard}, \cite{cheng_variation_obl})
\end{remark}

\subsection{Motion of an ellipsoid with stochastic flattening}

\subsubsection{Reminder about stochastic differential equations}

We remind basic properties and definition of stochastic differential equations in the sense of It\^o. We refer to the book \cite{oksendal} for more details.\\

A {\it stochastic differential equation} is formally written (see \cite[Chapter V]{oksendal}) in differential form as  
\begin{equation}
	\label{stocequa}
	dX_t = \mu (t,X_t)dt+\sigma(t,X_t)dB_t ,
\end{equation}
which corresponds to the stochastic integral equation
\begin{equation}
	\label{stocintegral}
	X_t=X_0+\int_0^t \mu (s,X_s)\,ds+\int_0^t \sigma (s,X_s)\,dB_s ,
\end{equation}
where the second integral is an It\^o integral (see \cite[Chapter III]{oksendal}) and $B_t$ is the classical Brownian motion (see \cite[Chapter II, p.7-8]{oksendal}).\\

An important tool to study solutions to stochastic differential equations is the {\it multi-dimensional It\^o formula} (see \cite{oksendal},Chap.III,Theorem 4.6) which is stated as follows : \\

We denote a vector of It\^o processes by $\mathbf{X}_t^\mathsf{T} = (X_{t,1}, X_{t,2}, \ldots, X_{t,n})$ and we put $\mathbf{B}_t^\mathsf{T} = (B_{t,1}, B_{t,2}, \ldots, B_{t,n})$to be a $n$-dimensional Brownian motion (see \cite{karatzas},Definition 5.1,p.72),  $d\mathbf{B}_t^\mathsf{T} = (dB_{t,1}, dB_{t,2}, \ldots, dB_{t,n})$. We consider the multi-dimensional stochastic differential equation defined by (\ref{stocequa}). Let $f$ be a $\mathcal{C}^2(\mathbb{R}_+ \times \mathbb{R},\mathbb{R})$-function and $X_t$ a solution of the stochastic differential equation (\ref{stocequa}). We have 
\begin{eqnarray}
	df(t,\mathbf{X}_t) = \frac{\partial f}{\partial t} dt + (\nabla_\mathbf{X}^{\mathsf T} f) d\mathbf{X}_t + \frac{1}{2} (d\mathbf{X}_t^\mathsf{T}) (\nabla_\mathbf{X}^2 f) d\mathbf{X}_t,
\end{eqnarray}
where $\nabla_\mathbf{X} f = \partial f/\partial \mathbf{X}$ is the gradient of $f$ w.r.t. $X$, $\nabla_\mathbf{X}^2 f = \nabla_\mathbf{X}\nabla_\mathbf{X}^\mathsf{T} f$ is the Hessian matrix of $f$ w.r.t. $\mathbf{X}$, $\delta$ is the Kronecker symbol and the following rules of computation are used : $dt dt = 0$, $dt dB_{t,i}  = 0$, $dB_{t,i} dB_{t,j} = \delta_{ij} dt$.

\subsubsection{Stochastic variation of the flattening}

Let $c_t$ be a stochastic process expressed as
\begin{equation}
	dc_t = f(t,c_t)dt + g(t,c_t)dB_t \label{dc_t}
\end{equation}
where $f,g \in \mathcal{C}^2(\mathbb{R}\times\mathbb{R},\mathbb{R})$. \\

\noindent{\bf Consequence of assumption (H1)} : Applying the It\^o formula on the volume formula (\ref{V})
\begin{equation}
	a_t^2=\frac{3V_\mathcal{E}}{4\pi} \frac{1}{c_t},
\end{equation}
we obtain
\begin{equation}
	d(a_t^2) = \frac{3V_\mathcal{E}}{4\pi} \left( -\frac{1}{c_t^2}dc_t + \frac{1}{c_t^3}(dc_t)^2 \right).
\end{equation}
Thus using the expression of $dc_t$ we obtain the following lemma :
\begin{lemma}
	Under assumption (H1) the variation of $a$ is given by
	\begin{equation}
		d(a_t^2) = \frac{3V_\mathcal{E}}{4\pi} \left[\left( -\frac{f(t,c_t)}{c_t^2} + \frac{g(t,c_t)^2}{c_t^3}  \right)dt - \frac{g(t,c_t)}{c_t^2} dB_t \right]. \label{da2}
	\end{equation}
\end{lemma}
We can now determine the variation of the inertia matrix coefficients $I_1$ and $I_3$. \\

\noindent{\bf Consequence of assumption (H2)} : Applying the It\^o formula on the expression of $I_1$ and $I_3$ leads to

\begin{lemma}
	Under assumption (H2) the variation of $I_1$ and $I_3$ are given by
	\begin{equation}
		dI_3 = \frac{3 M_\mathcal{E} V_\mathcal{E}}{10\pi}\left[\left( -\frac{f(t,c_t)}{c_t^2} + \frac{g(t,c_t)^2}{c_t^3}  \right)dt - \frac{g(t,c_t)}{c_t^2} dB_t \right]
	\end{equation}
	and
	\begin{align*}
		dI_1 = \frac{M_\mathcal{E}}{5}\bigg[& \left(-\frac{3V_\mathcal{E}}{4\pi}\frac{f(t,c_t)}{c_t^2} + g^2(t,c_t)\left(1+\frac{3V_\mathcal{E}}{4\pi c_t^3}\right) + 2 c_t f(t,c_t)\right)dt \nonumber \\
		&+ g(t,c_t)\left(2 c_t-\frac{3V_\mathcal{E}}{4\pi c_t^2} \right)dB_t \bigg] .
	\end{align*}
\end{lemma}

\subsubsection{Stochastic equations of motion}

In order to formulate the equations of motion of $\mathcal{E}$ with a stochastic flattening, we first rewrite the equations of motion (\ref{eqmotion}) in differential form, which is the natural form for the stochastic process, in order to have coherent form of writing :
\begin{equation}
	dL_i = l_i(\mathbf{I},\mathbf{\Omega})dt ,
\end{equation}
where $l_i(\mathbf{I},\mathbf{\Omega})$ are the same as previous. Taking into account our stochastic variation of the flattening we get the full set of the stochastic equations of motion for $\mathcal{E}$ as
\begin{align}
	dL_i & = l_i({\bf I},\boldsymbol \Omega)dt \label{Li},\\
	dI_i & =  h_i(c_t)dt + m_i(c_t)dB_t \nonumber \label{Ii}, 
\end{align}
for $i=1,2,3$ where

\begin{equation}
	\begin{array}{lll}
		h_1(c_t)&=&\frac{M_\mathcal{E}}{5} \left(-\frac{3V_\mathcal{E}}{4\pi}\frac{f(t,c_t)}{c_t^2} + g^2(t,c_t)\left(1+\frac{3V_\mathcal{E}}{4\pi c_t^3}\right) + 2 c_t f(t,c_t)\right), \\
		h_3(c_t)&=&\frac{3 M_\mathcal{E} V_\mathcal{E}}{10\pi}\left( -\frac{f(t,c_t)}{c_t^2} + \frac{g(t,c_t)^2}{c_t^3} \right), \\
		h_2(c_t)&=&h_3(c_t),
	\end{array}
\end{equation}
\begin{equation}
	\begin{array}{lll}
		m_1(c_t)&=&\frac{M_\mathcal{E}}{5} g(t,c_t)\left(2 c_t-\frac{3V_\mathcal{E}}{4\pi c_t^2} \right), \\
		m_3(c_t)&=&-\frac{3 M_\mathcal{E} V_\mathcal{E}}{10\pi}\frac{g(t,c_t)}{c_t^2}, \\
		m_2(c_t)&=&m_3(c_t).
	\end{array}
\end{equation}

A stochastic version of the Euler-Liouville equation induced by the stochastic flattening is then obtained as follows : As we consider only variation of the flattening, we have a rotational symmetry during the deformation. Hence, we have $L_i = I_i \Omega_i$  or equivalently $\Omega_i = \frac{L_i}{I_i}$ for $i=1,2,3$. Thus, using the It\^o formula for each component of $\mathbf{\Omega}$, we obtain :

\begin{definition}
We call Stochastic Euler-Liouville equations for an ellipsoid with a stochastic flattening the following equations
\begin{equation}
	\begin{array}{lll}
		d\Omega_i &=& \left(\frac{l_i({\bf I},\boldsymbol \Omega)}{I_i} - \frac{\Omega_i}{I_i}h_i(c_t)+\frac{\Omega_i}{I_i^2}m_i^2(c_t) \right)dt - \frac{\Omega_i}{I_i}m_i(c_t)dB_t, \\
		dI_i & = & h_i(c_t)dt + m_i(c_t)dB_t, \\
		dc_t &= & f(t,c_t)dt + g(t,c_t)dB_t
	\end{array}\label{euler_liouville_sto}
\end{equation}
for $i=1,2,3$.
\end{definition}

As we can see, there exist a drift or secular variation in the rotation vector, represented by the term $\di \frac{\Omega_i}{I_i^2}m_i^2(c_t)$ which is induced by the stochastic nature of the variations considered. If one would like to interpret the rotation vector in terms of the so-called, \emph{Euler angles}, one would observe a secular variation in the angles.

\begin{remark}
Let us remark that Stochastic Euler-Liouville equations are also valid if one would like to consider directly variations on the inertia matrix coefficient $I_1$ and $I_3$. In that case we would have these equations written as
\begin{equation}
\begin{array}{lll}
d\Omega_i &=& \left(\frac{l_i({\bf I},\boldsymbol \Omega)}{I_i} - \frac{\Omega_i}{I_i}h_i(t,I_i)+\frac{\Omega_i}{I_i^2}m_i^2(t,I_i) \right)dt - \frac{\Omega_i}{I_i}m_i(t,I_i)dB_t, \\
dI_i & = & h_i(t,I_i)dt + m_i(t,I_i)dB_t,
\end{array}
\end{equation}
for $i=1,2,3$. Such a case is interesting when one want to model variation of the shape of a body in term of the inertia matrix coefficients, the dynamical flattening $H$ or the zonal harmonic $J_2$ as in \cite{yoder}. In that precise case, instead of considering boundedness variation of the flattening $c_t$, one can formulate the assumption (H3) in term of the \emph{invariance of the trace of inertia matrix} using the result in \cite{rochester}.
\end{remark}

\subsubsection{Variation of the dynamical flattening H and the zonal harmonic $J_2$}
\label{variation_seculaire}
\begin{proposition}
\label{H_drift}
Under assumptions (H1) and (H2),  the variation of the dynamical flattening $H$ is given by
\begin{equation}
dH=-\frac{2\pi}{V_\mathcal{E}}c_t^2 dc_t-\frac{2\pi}{V_\mathcal{E}}c_t g(t,c_t)^2dt
\end{equation}
or equivalently by
\begin{equation}
H=H_0-\frac{2\pi}{V_\mathcal{E}} \int_{c_0}^{c_t} c_s^2 dc_s-\frac{2\pi}{V_\mathcal{E}} \int_{0}^{t}c_s g(s,c_s)^2ds.
\end{equation}
\end{proposition}

\begin{proof}
\begin{equation}
\label{dH_Ii}
dH=\frac{I_1dI_3-I_3dI_1}{I3^2}+\frac{I_3 dI_1dI_3-I_1(dI_3)^2}{I_3^3}.
\end{equation}
Using the expression of the variation of $I_1$ and $I_3$, we obtain
\begin{equation}
dH=\left[\frac{I_1h_3-h_1I_3}{I_3^2}+\frac{m_3(I_3m_1-I_1m_3)}{I_3^3}\right]dt+\frac{I_1m_3-m_1I_3}{I_3^2}dB_t.
\end{equation}
From expressions of $I_1$, $I_3$,$h_1$,$h_3$,$m_1$ et $m_3$, we obtain 
\begin{equation}
dH=-\frac{2\pi}{V_\mathcal{E}}c_t^2 dc_t-\frac{2\pi}{V_\mathcal{E}}c_t (dc_t)^2.
\end{equation}
As $(dc_t)^2=g(t,c_t)^2dt$, we obtain the result.
\end{proof}

From the expression of the zonal harmonic $J_2$ and the variation of the dynamical flattening, we obtain its variation:
\begin{lemma}
\label{j2_drift}
Under assumptions (H1) and (H2),  the variation of the zonal harmonic $J_2$ is given by
\begin{equation}
dJ_2=\frac{4\pi}{5V_\mathcal{E}}c_t^2 dc_t+\frac{4\pi}{5V_\mathcal{E}}c_t g(t,c_t)^2dt
\end{equation}
or equivalently by
\begin{equation}
J_2=J_{2,0}+\frac{4\pi}{5V_\mathcal{E}} \int_{c_0}^{c_t} c_s^2 dc_s+\frac{4\pi}{5V_\mathcal{E}} \int_{0}^{t}c_s g(s,c_s)^2ds.
\end{equation}
\end{lemma}

\begin{remark}
The term $\di \int_{0}^{t}c_s g(s,c_s)^2ds$ in the variation of the dynamical flattening $H$ or the zonal harmonic $J_2$, is exactly the consequence of the stochastic nature of the variation of the flattening. It induces a drift which could be found when studying the long time behavior of quantity depending on the term $H$ or $J_2$, such as the \emph{length of the day}.
\end{remark}

From the previous proposition and remark, one can see a non negligible consequence of such a stochastic model. Indeed, considering only deterministic variations, there is no chance to obtain the drift induced by the It\^o formula and by consequence, it is impossible to understand why there exist long time drift for example, in the \emph{length of the day}.

\begin{remark}
From a practical point of view, one has to study this extra term to model the stochastic process governing the variation of the flattening.
\end{remark}

\subsubsection{Admissible stochastic deformations}
\label{adm_stoc_def}
A stochastic process has in general unbounded variations. We have to take precautions when considering stochastic fluctuations of the flattening. Indeed, the assumption (H3) mainly concerns the purely stochastic part and has to be interpreted as a way to have not a noise which ``explodes''. This is the main difference with the deterministic case. Even if it has been showed recently (see \cite{cheng_deceleration}) that there exists a secular variation in the zonal harmonic $J_2$, it is not incompatible with the coupling of a stochastic variation in the flattening which has bounded variations. Indeed, bounded variations of the flattening also induce a drift in the zonal harmonic $J_2$ (see Lemma \ref{j2_drift}). \\

The main constraint on the deformation in the stochastic case comes from the boundedness assumption.  :

\begin{definition}
	If $c_t$ satisfies the condition $\mathbb{P}\left(c_0 + d_{min} \leq c_t \leq c_0 + d_{max}\right)=1$ for $t\ge 0$ then, we say that $c_t$ is an admissible stochastic deformation where $\mathbb{P}$ is the underlying probability measure.
\end{definition}

In order to characterize admissible stochastic deformations, we use the stochastic invariance theorem (see (\cite{milian}) :

\begin{theorem}\label{mainB}
	Let $a,b\in\mathbb{R}$ such that $b>a$ and $dX(t) =  f(t,X(t) ) dt + g(t,X(t) ) dB_t$ a stochastic process. Then, the 
	set 
	$$
	K:=\{x\in\mathbb{R} :\ a\leq x\leq b\}
	$$ 
	is invariant for the stochastic process $X(t)$ if and only if
	\begin{eqnarray}
		f (t,a) &\geq & 0, \nonumber \\
		f (t,b) &\leq & 0,  \nonumber \\
		g(t,x) &= &0 \quad \text{for } x\in\{a,b\},\nonumber
	\end{eqnarray}
	for all $t\geq 0$.
\end{theorem}

As a consequence, we have :

\begin{lemma}[Characterization of admissible stochastic deformations]
	Let $c_t$ satisfying $dc_t = f(t,c_t)dt + g(t,c_t)dB_t$ then, $c_t$ is an admissible deterministic variation if and only if
	\begin{align*}
		f(t,c_0 + d_{min}) & \geq 0, \\
		f(t,c_0 + d_{max}) & \leq 0 \ , \quad \forall t \geq 0, \\
		g(t,c_0 +d_{min}) & = g(t,c_0 + d_{max}) = 0 \ , \quad \forall t \geq 0 .
	\end{align*}
\end{lemma}

\subsubsection{A stochastic Toy-model}

In order to perform numerical simulations, we introduce an ad-hoc deformation defined by
\begin{align}
	f(x) & = \alpha\cos(\gamma t)(x-(c_0 + d_{min}))((c_0+d_{max})-x) \ , \quad e,\alpha \in \mathbb{R}^+ ,\\
	g(x) & = \beta(x-(c_0 + d_{min}))((c_0+d_{max})-x) \ , \beta \in \mathbb{R}^+ ,
\end{align}
where $\beta$ is a real number. The function $g$ is designed to reproduce the observed stochastic behavior of the flattening of the Earth. However, as pointed out in the introduction, we do not intend to produce an accurate model but mainly to study if such a model using stochastic processes leads to a good agreement on the shape of the polar motion. \\

The major axis $c_t$ satisfies the stochastic differential equation
\begin{align}
dc_t = &\alpha\cos(\gamma t)(c_t-(c_0 + d_{min}))((c_0+d_{max})-c_t)dt  \nonumber \\
&+ \beta(c_t-(c_0 + d_{min}))((c_0+d_{max})-c_t)dB_t \label{procInv} .
\end{align}

\section{Simulations of the Toy-model}

\subsection{Initial conditions}

All the simulations are done under the following set of initial conditions :
\begin{itemize}
	\item The semi-major axis of $\mathcal{E}$ : $a_0=1$, $c_0=\sqrt{\frac{298}{300}}$.
	\item Mass : $M_\mathcal{E}=1$.
	\item Volume : $V_\mathcal{E}=1$. 
	\item Rotation vector $\boldsymbol{\Omega}$ is chosen in the principal axis as $\boldsymbol{\Omega} = \left(5\times 10^{-7},0,1 \right)^\mathsf{T}$. 
	\item Upper variation $d_\text{max}=a_0-c_0$.
	\item Lower variation $d_\text{min}=-d_\text{max}$. 	
	\item Perturbation coefficients : $\alpha=10^{-3}$ and $\beta=10^{-4}$ with $\gamma =10$.
	
\end{itemize}
These initial conditions correspond to the Earth which rotate around its axis in about 300 days, oscillating with a circle of radius about 3 meters (see \cite{bizouard}, \cite{goldstein}). The perturbation coefficients and also the upper and lower variations are arbitrary. The reader can test different values of the initial conditions using the open-source Scilab program made by F. Pierret (see \cite{pierret3})

\subsection{Numerical scheme and the invariance property}
\label{num_inv}
As we do not perform simulations over a long time, we can use in the deterministic case the Euler scheme and in the stochastic case the Euler-Maruyama scheme. However, in each case a difficulty appears which is in fact present in many other domains of modeling (see for example \cite{cps1} and \cite{cps2}), namely the respect of the invariance condition under discretization. Indeed, even if the continuous model satisfies the invariance condition leading to an admissible deformation, the discrete quantity can sometimes produce unrealistic values leading to, for example, negative values of the major axis. Thanks to an appropriate choice (see \cite{pierret_nsem}) of the time step, it is possible (under some conditions) to obtain a numerical scheme satisfying the invariance property (with a probability which can be as close as we want to one in the stochastic case).\\

In the following, we denote by $h \in \mathbb{R}^+$ the time increment of the numerical scheme. For $n\in \mathbb{N}$, we denote by $t_n$ the discrete time defined by $t_n = nh$ and by $X_n$ the numerical solution compute at time $t_n$ with time step $h=10^{-4}$. In the simulations, the value of $a_0$ can be seen as the \emph{Earth's equatorial radius}, which allow us to display the variations in the oscillation with magnitude of order few centimeters. The simulations are performed over 1 day and 7 days in order to exhibit the random phenomena linked to the period of few days.

\subsubsection{Deterministic case}
In order to perform numerical simulations we use the Euler scheme. Let $X_t$ a smooth function such that
\begin{equation}
X_t=X_0+\int_0^t f(s,X_s)\,ds
\end{equation}
where $f \in \mathcal{C}^2(\mathbb{R}\times\mathbb{R},\mathbb{R})$. The associated Euler scheme associate is given by
\begin{equation}
X_{n+1} = X_n + f(t_n,X_n) h.
\end{equation}
Considering the Euler scheme associate with the stochastic Euler-Liouville equations \eqref{euler_liouville_sto} and the toy-model for the flattening \eqref{procInv} with only a deterministic variation, we display the difference between $c_t$ and its initial value $c_0$ in Figure \ref{fig:ctdet1}. We display also the difference between the zonal harmonic $J_2$ and its initial value $J_{2,0}$ in Figure \ref{fig:j2det1}. The difference is made to see the variation order as, for example, in \cite{cheng_variation_obl}.

\begin{figure}[ht!]
\resizebox{\textwidth}{!}{
	\includegraphics{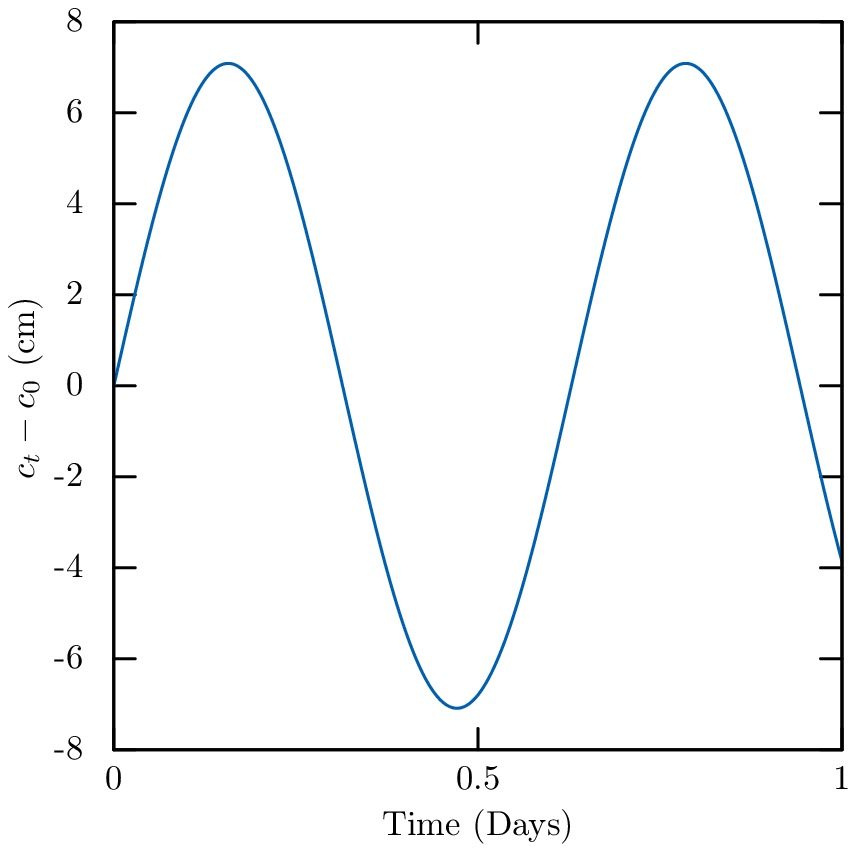}
	\includegraphics{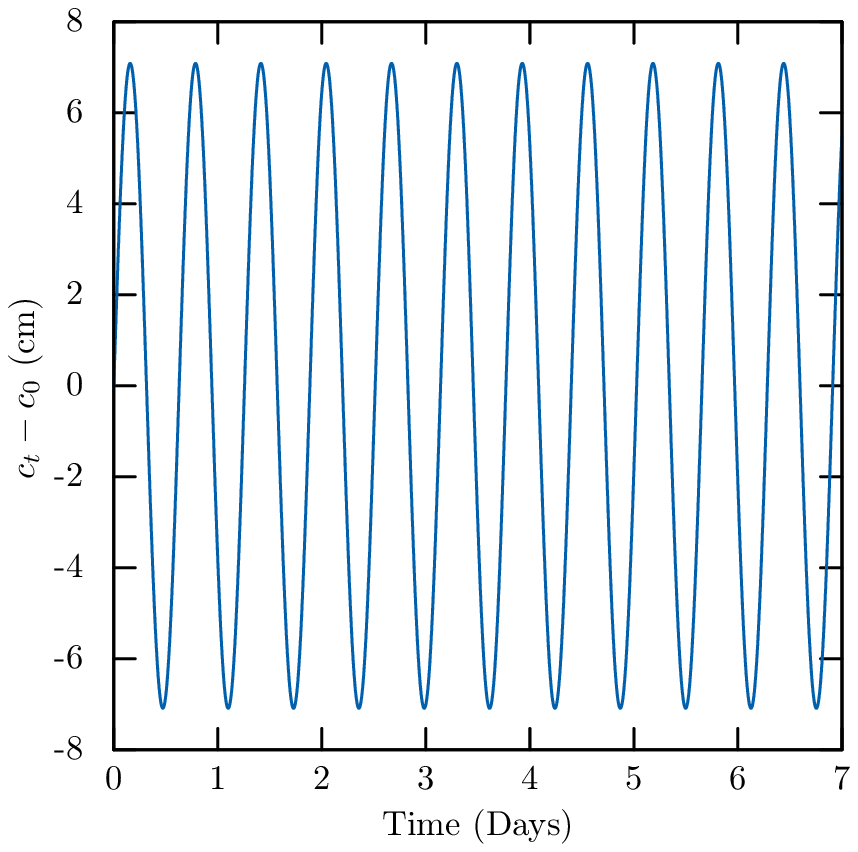}
}
\caption{Semi-major axis $c_t$}
\label{fig:ctdet1}
\end{figure}

\begin{figure}[ht!]
\resizebox{\textwidth}{!}{
	\includegraphics{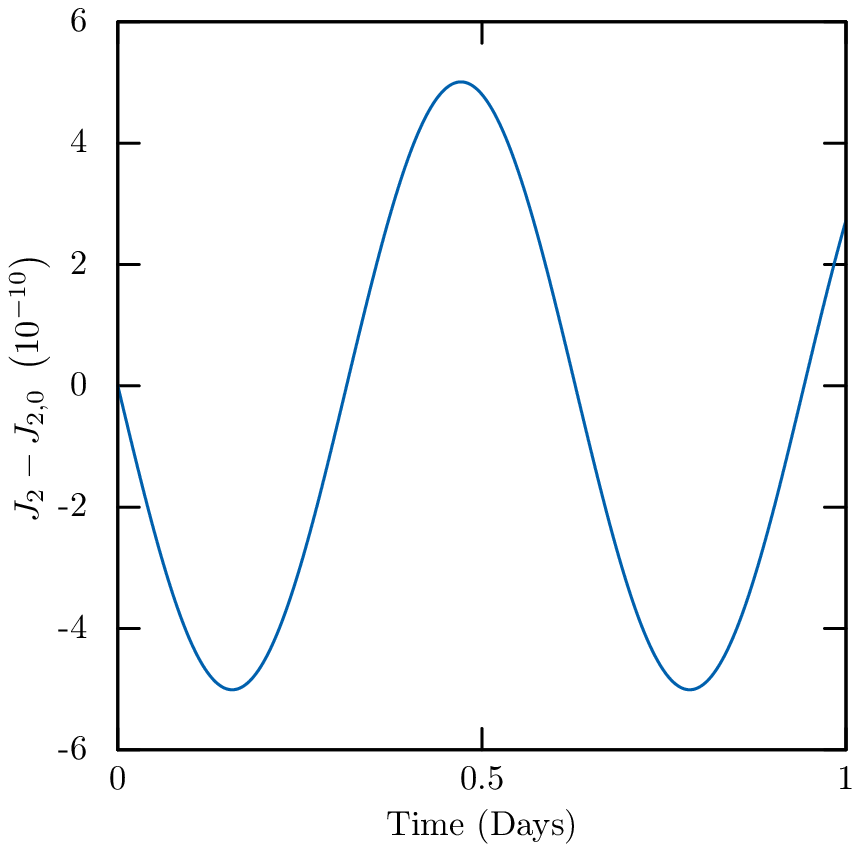}
	\includegraphics{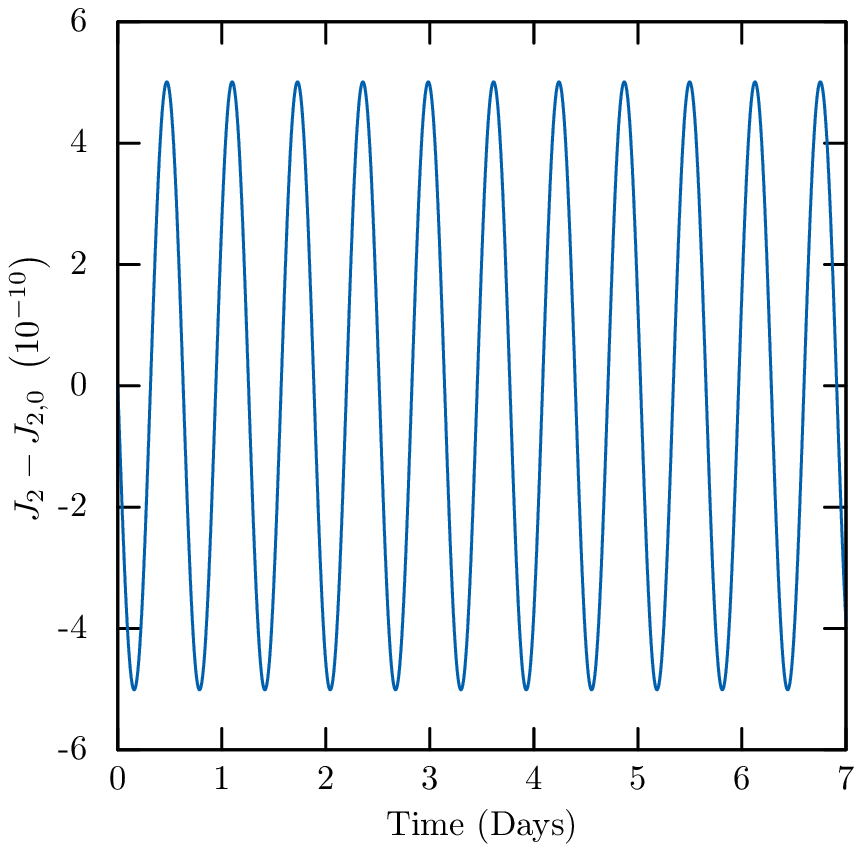}
}
\caption{Second zonal harmonic $J_2$}
\label{fig:j2det1}
\end{figure}

As it has been precised in the introduction, this model intends to introduce the deterministic part of the stochastic deformation. Such a model has to be replaced by the actual deterministic models, for example, the part with the well known periodic variations (see \cite{bizouard}, \cite{lambeck_1989}, \cite{sidorenkov2009}).

\subsubsection{Stochastic case}
In order to do numerical simulations we use the Euler-Maruyama scheme which is the stochastic counterpart to the Euler-Liouville scheme for deterministic differential equations (see \cite{higham}, \cite{kloeden1}). Let $X_t$ be a stochastic process written as

\begin{equation}
	X_t=X_0+\int_0^t f(s,X_s)\,ds+\int_0^t g(s,X_s)\,dB_s
\end{equation}
where $f,g \in \mathcal{C}^2(\mathbb{R}\times\mathbb{R},\mathbb{R})$. The Euler-Maruyama scheme is given by

\begin{equation}
	X_{n+1} = X_n + f(t_n,X_n) h + g(t_n,X_n) \Delta B_n,
\end{equation}
where $\Delta B_n$ is a Brownian increment which is normally distributed with mean zero and variance $h$ for all $n\ge 0$.\\

Considering the Euler-Maruyama scheme associate with the stochastic Euler-Liouville equations \eqref{euler_liouville_sto} and the toy-model for the flattening \eqref{procInv}, we display the difference between $c_t$ and its initial value $c_0$ in Figure \ref{fig:ctsto1}, and we display the zonal harmonic $J_2$ and its initial value $J_{2,0}$ in Figure \ref{fig:j2sto1}.

\begin{figure}[ht!]
\resizebox{\textwidth}{!}{
	\includegraphics{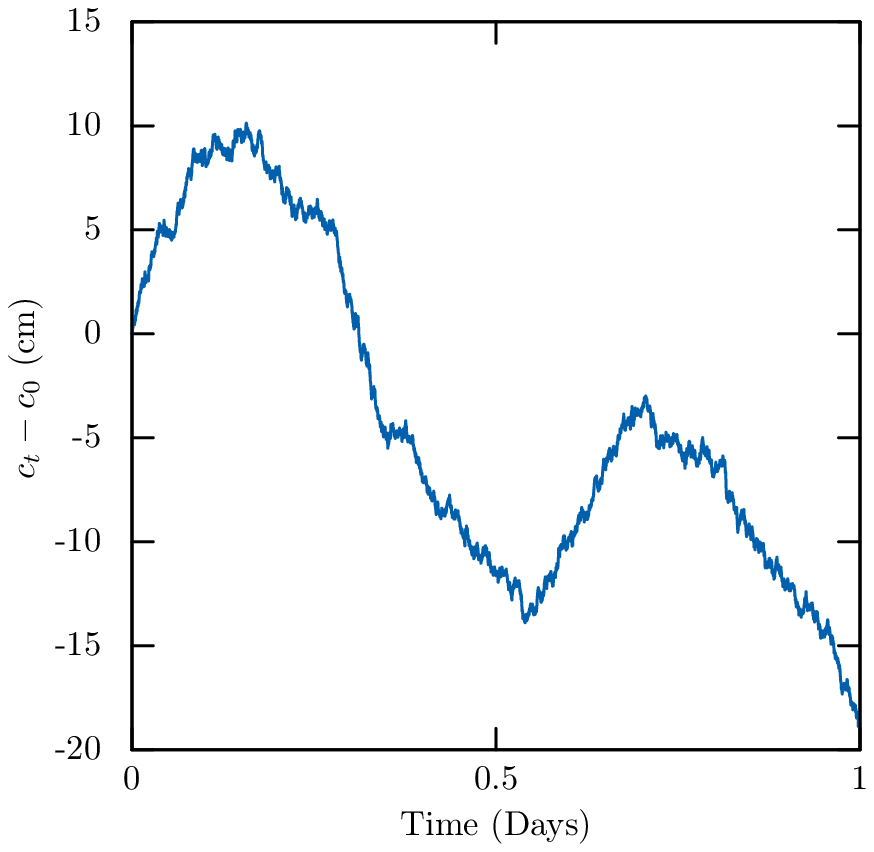}
	\includegraphics{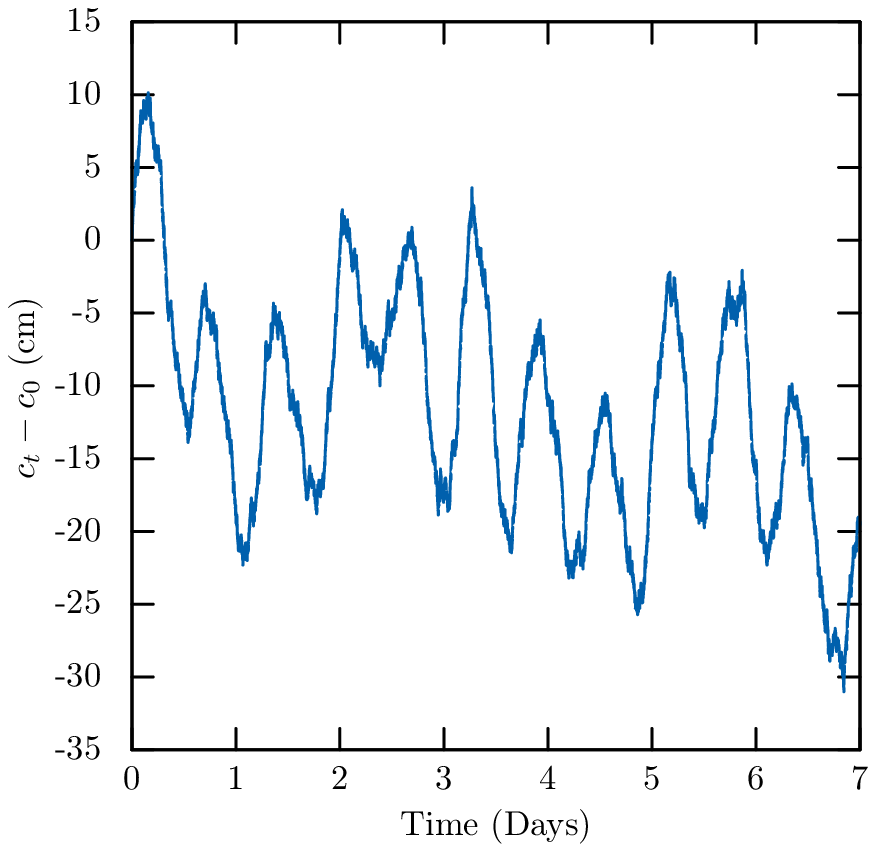}
}
\caption{Semi-major axis $c_t$}
\label{fig:ctsto1}
\end{figure}

\begin{figure}[ht!]
\resizebox{\textwidth}{!}{
	\includegraphics{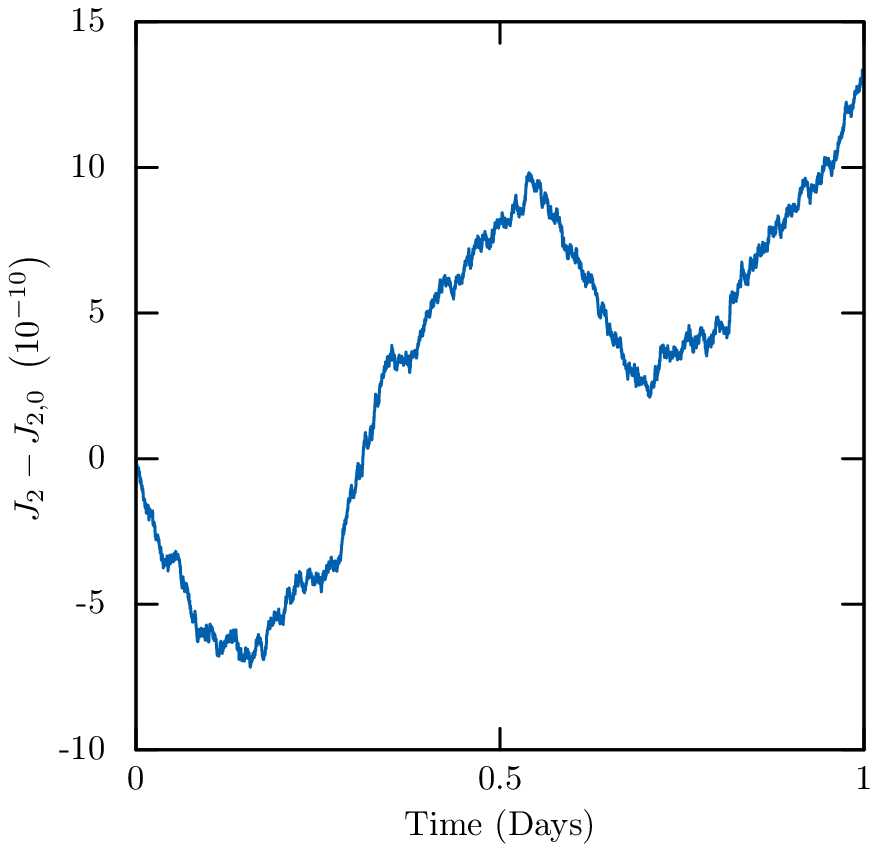}
	\includegraphics{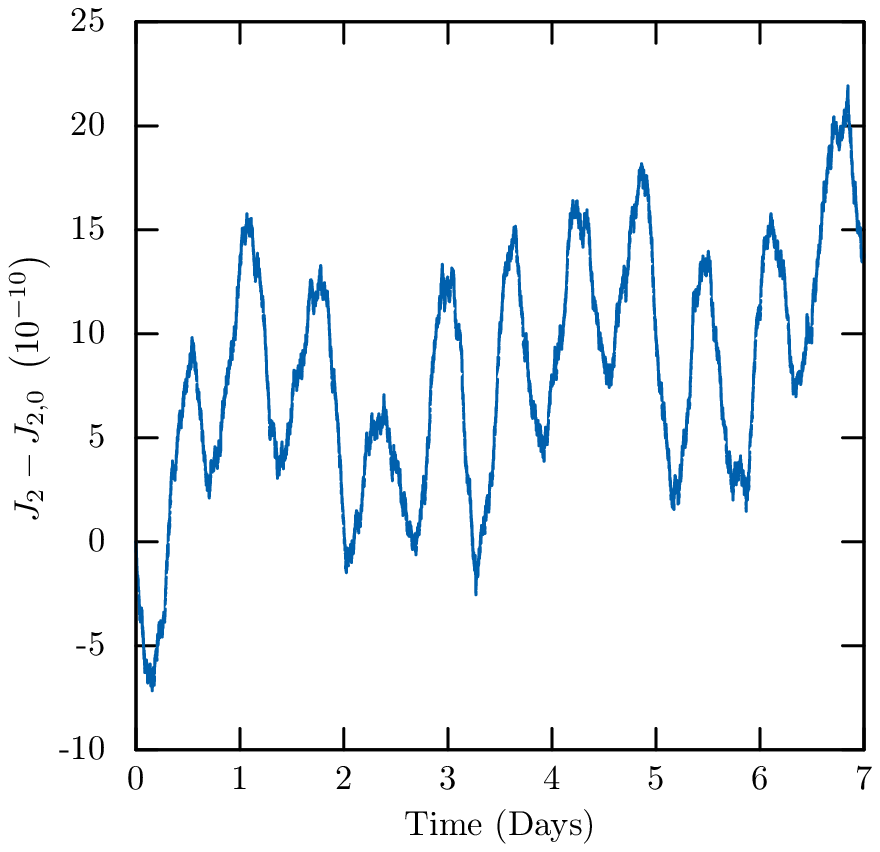}
}
\caption{Second zonal harmonic $J_2$}
\label{fig:j2sto1}
\end{figure}

Considering stochastic variations of the flattening, we can see that over short periods of time, there exists similarities between the general shape of the flattening curve obtained using simulations of the stochastic toy-model and the observational curves (see \cite{cheng_deceleration}, \cite{cheng_variation_obl}). It shows that the model seems to capture a part of the random effects which could be inside the real observations. \\

Nevertheless, the interpretation of the numerical illustrations with the observational data has to be understand as follows. Such a simple model does not intend to reproduce the actual behavior of the Earth but to show the strategy to interpret the random effects observed. Indeed, when analyzing the real data of the Earth, such as the second zonal harmonic $J_2$, one has to remove all the well known terms like the periodic ones and all others we know exactly quantify. Then, in the ``residual'' data, the observed noise has to be interpreted. Even if, we remove the noise induced by the measure process, this actual noise is very badly understand from the physics or the geological point of view. This exactly what we are showing in the numerical example. If one has periodic terms in its data and remove it, we would obtain a noisy signal which can be exactly identified and modeled as a stochastic variation of the flattening. \\

As precised in the introduction, stochastic variations in the flattening can explain why it is so difficult to predict the rotation motion over few days and also why the filtering methods seems to work but without providing the physical meaning and origin.

\section{Conclusion and perspectives}
This study is a first attempt to take into account stochastic variation of the shape of a body on its rotation through the geometric flattening. It shows that, if there exists a quantification of geophysical mechanisms or an intrinsic description of the ellipsoid on the stochastic variations of the flattening then, this work gives a method to deal with admissible stochastic variations under the assumptions (H1), (H2) et (H3). Thanks to this physical quantification, it allows exhibiting one of the many candidates of the Earth's stochastic rotation dynamics. The results encourage working in this direction. \\

Other mechanisms, such as the non-rigidity of the Earth, induce a major part in the rotation behavior (see \cite{bizouard}, \cite{lambeck1988}). With such a consideration, it follows that Earth's rotation axis is not described in the principal axis and, in consequence, the inertia matrix is not necessarily diagonal. Of course it is possible to adapt all the theoretical and numerical results of this work in such a situation and moreover with a body having a general shape. \\

Obviously the hypothesis on the stochastic nature of the deformation which should be a diffusion process seems to be too restrictive. Indeed, the real data suggest there is sometimes noise coloration (see \cite{bizouard}, \cite{markov2009}). Of course it is also possible to adapt all the results of this work with colored noise (see \cite{hanggi_jung} and \cite{riecke} for example for a short introduction to colored noise using the Ornstein-Uhlenbeck process and its white noise limit). Testing such a model with the real data and colored noise and also the actual model of the main perturbation of Earth's rotation such as the oceanic and atmospheric excitation will be the subject of a future paper. \\

One of the applications of this work, is to model a two-body problem perturbed by these stochastic variations of the flattening on its orbitals elements. In consequence, the satellite dynamics used to acquiring data can be investigated when doing a comparison between the data and the Earth's dynamics. Another application, suggested by one of the referee, is to consider such a stochastic approach to the deformation of the Earth and the Moon on the Moon's rotation vector during its formation process. Indeed, an idea due to R.W. Ward  (see \cite{ward}) suggests that the Cassini states, which were present, allowed for the Moon's rotation vector to undergo a radical 90 degree flip. It would be interesting to understand, using stochastic deformations, if such a reorientation of the Moon's rotation vector were possible. It could validate a fundamental question relating to the origin and evolution of the Earth-Moon system.

\section{Acknowledgment}
We would like to thank the reviewers for their insightful comments on the paper which led us to an improvement of this work. We would also like to thank S\'ebastien Lambert for his careful proofreading and discussions.

\end{document}